\documentclass[format=sigconf]{acmart}

\usepackage{bm}
\usepackage{amsmath,amsthm,amsmath}
\usepackage{algorithm, pseudocode}
\usepackage{mathrsfs}
\usepackage{enumitem}
\usepackage{xcolor}

\usepackage{multirow}

\def\headline#1{\hbox to \hsize{\hrulefill\quad\lower.3em\hbox{#1}\quad\hrulefill}}
\def\hline#1{\hbox to \hsize{\line(1,0){10}\quad \lower.3em\hbox{$#1$}\quad \line(1,0){10}}}

\definecolor{ao(english)}{rgb}{0.0, 0.5, 0.0}

\date{}

\newcommand{\X}{\underline{X}}
\newcommand{\W}{\mathcal{W}}

\newcommand{\abb}[5]{%
\setlength{\arraycolsep}{0.4ex}%
\begin{array}{rcccc}%
#1 &:\,& #2 & \,\,\longrightarrow\,\, & #3 \\[0.5ex]%
     & & #4 & \longmapsto & #5%
\end{array}%
}

\newtheorem{definition}{Definition}
\numberwithin{definition}{section}
\newtheorem{theorem}[definition]{Theorem}

\newtheorem{proposition}[definition]{Proposition}
\newtheorem{lemma}[definition]{Lemma}

\newtheorem{example}[definition]{Example}
\newtheorem{notation}{Notation}[section]

\def\K{\ensuremath{\mathbb{Q}}}

\def\R{\ensuremath{\mathbb{R}}}

\def\Q{\ensuremath{\mathbb{Q}}}
\def\N{\ensuremath{\mathbb{N}}}

\DeclareBoldMathCommand{\a}{a}
\DeclareBoldMathCommand{\c}{c}
\DeclareBoldMathCommand{\e}{e}
\DeclareBoldMathCommand{\f}{F}
\DeclareBoldMathCommand{\g}{g}
\DeclareBoldMathCommand{\h}{h}
\DeclareBoldMathCommand{\x}{z}
\DeclareBoldMathCommand{\z}{z}
\DeclareBoldMathCommand{\v}{v}

\def\softO{\ensuremath{{O}{\,\tilde{ }\,}}}

\DeclareMathOperator{\Comp}{Comp}
\DeclareMathOperator{\comp}{comp}

\DeclareMathOperator{\sign}{sign}
\DeclareMathOperator{\thom}{Thom}
\DeclareMathOperator{\der}{Der}
\def\f{\ensuremath{F}}
\def\scrQ{\ensuremath{\mathscr{Q}}}

\DeclareBoldMathCommand{\p}{p}
\DeclareMathOperator{\RM}{RM}
\DeclareMathOperator{\Com}{Comp}
\DeclareMathOperator{\length}{length}
\DeclareMathOperator{\CompMax}{CompMax}

\title[Connectivity in Symmetric Semi-Algebraic Sets]{Connectivity in Symmetric Semi-Algebraic Sets} 
\thanks{C.\ Riener and R.\ Schabert was supported by the Troms\o{} Research Foundation grant 17MATCR.  T. X. Vu was partially supported by the ANR-FWF grant 10.55776/I6130. 
}

\author[C. Riener, R.  Schabert, T. X. Vu]{Cordian Riener$^{a}$, Robin  Schabert$^{a}$, Thi Xuan Vu$^{a, b}$} 
\affiliation{%
  \institution{$^a$Department of Mathematics and Statistics, UiT The Arctic University of Norway, Tromsø, Norway}
  \institution{$^b$Institute for Algebra, Johannes Kepler University, Linz, A4040, Austria}  
  \state{}
  \postcode{}
  \country{}
}

\copyrightyear{2024}
\acmYear{2024}
\setcopyright{rightsretained}
\acmConference[ISSAC '24]{International Symposium on Symbolic and Algebraic Computation}{July 16--19, 2024}{Raleigh, NC, USA}
\acmBooktitle{International Symposium on Symbolic and Algebraic Computation (ISSAC '24), July 16--19, 2024, Raleigh, NC, USA}
\acmDOI{10.1145/3666000.3669687}
\acmISBN{979-8-4007-0696-7/24/07}

\begin{document}

\begin{abstract}
A semi-algebraic set is a subset of the real space defined by
polynomial equations and inequalities. In this paper, we consider the 
problem of deciding whether two given points in a semi-algebraic set
are connected. We restrict to the case when all equations and
inequalities are invariant under the action of the symmetric
group and of degree at most $d<n$, where $n$ is the number of variables. Additionally, we assume that the two points are in the same fundamental domain of the action of the symmetric group, by assuming that the coordinates of two given points are sorted in non-decreasing order. We construct and analyze an algorithm that solves this
problem, by taking advantage of the group action, and has a complexity being polynomial in $n$.   
\end{abstract}

\maketitle 
 
\section{Introduction}

\subsection{Motivations}
A motivation for determining the connectivity of semi-algebraic sets
is derived from applications in robot motion planning
\cite{schwartz1983piano}. This is equivalent to ascertaining whether
two corresponding points within the free space belong to the same
connected component of the free space. More recently, practical
applications of roadmap algorithms to kinematic singularity analysis
were reported in Capco et al. \cite{capco2020robots,
  capco2023positive}, highlighting the significance of extending
roadmap algorithms beyond motion planning. Simultaneously, the growing
interest in these algorithms is evident as they are adapted to the
numerical side (see e.g. \cite{iraji2014nuroa}). This underscores the
importance of enhancing roadmap algorithms and the underlying
connectivity results.  

\subsection{Prior results}
\label{subsec:result}
The connectivity problem has received significant attention; however,
to the best of our knowledge, there is no previous work that takes
symmetry into account. 
\paragraph*{Computing roadmaps and deciding the connectivity.}
Pioneering work  of  Schwartz and Sharir \cite{schwartz1983piano}
proposed a solution using  Collins' cylindrical algebraic
decomposition method. The computational  complexity of their approach
is polynomial in the degree $d$ and the number of polynomials $s$, but
it is doubly exponential in the number  $n$ of variables. 

Canny introduced the concept of a roadmap for a semi-algebraic set and
presented an algorithm \cite{Canny}. Roadmaps offer a method for
counting connected components and determining whether two points
belong to the same connected component. Subsequent modifications
\cite{canny1993computing} were made to this algorithm, resulting in
the construction of a roadmap for a semi-algebraic set defined by
polynomials whose sign invariant sets give a stratification of $\R^n$
and whose complexity is $s^n (\log s) d^{O(n^4)}$.  In the case of an
arbitrary semi-algebraic set, he introduced perturbations to the
defining polynomials. This modification enables the algorithm to
determine whether two points belong to the same semi-algebraically
connected component with the same complexity. However, this algorithm
does not provide a path connecting the points. There is also a Monte
Carlo version of this algorithm, which has complexity  $s^n (\log s)
d^{O(n^2)}$.  

Grigor'ev and Vorobjov \cite{grigor1992counting,canny1992finding} proposed an algorithm
with complexity $(sd{})^{n^{O(1)}}$ for counting the number of connected
components of a semi-algebraic set. Additionally, Heintz, Roy, and
Solerno \cite{heintz1994single} and  Gournay and Risler
\cite{gournay1993construction} presented algorithms capable of
computing a roadmap for any semi-algebraic set with the same
complexity. Unlike Canny's algorithm, the complexities of these
algorithms are not separated into a combinatorial part (dependent on
$s$) and an algebraic part (dependent on $d$). Since the given
semi-algebraic set might have $(sd)^n$ different connected components,
Canny's algorithm's combinatorial complexity is nearly optimal. 

In \cite{basu2000computing}, a deterministic algorithm is presented,
constructing a roadmap for any semi-algebraic set within an algebraic
set of dimension $k$ with complexity $s^{k+1}d^{O(k^2)}$. This
algorithm is particularly relevant in robot motion planning, where the
configuration space of a robot is often embedded as a
lower-dimensional algebraic set in a higher-dimensional Euclidean
space. The algorithm's complexity is tied to the dimension of the
algebraic set rather than the ambient space, making it advantageous
for such scenarios. Its combinatorial complexity is nearly optimal,
utilizing only a fixed number of infinitesimal quantities, which
reduces the algebraic complexity to $d^{O(k^2)}$. Additionally, the
algorithm computes a semi-algebraic path between input points if they
lie in the same connected component, addressing the full scope of the
problem. 

None of the mentioned above algorithms has a cost lower than
$d^{O(n^2)}$ and none of them returns a roadmap of degree lower
than $d^{O(n^2)}$. Safey El Din and Schost \cite{safey2011baby} gave a
probabilistic algorithm that extended Canny's original approach to
compute a roadmap of a closed and bounded hypersurface of complexity
$(nd)^{O(n^{1.5})}$.
Later on in \cite{basu2014baby}, the same
authors, together with  Basu and Roy, showed how to obtain a
deterministic algorithm for computing a roadmap of a general real
algebraic set within a cost of $d^{O(n^{1.5})}$.  Basu and Roy in
\cite{basu2014divide}  used a divide-and-conquer strategy to divide
the current dimension by two at every recursive step, leading to a
recursion tree of depth $O(\log(n))$. Their deterministic algorithm
computes a roadmap for a hypersurface in time polynomial in
$n^{n\log^3(n)}d^{n\log^2(n)}$ while the output has size polynomial in
$n^{n\log^2(n)}d^{n\log(n)}$. It is important to note that this algorithm is not polynomial in its output size. 
 Despite
not assuming smoothness on the hypersurface, the algorithm is capable
of handling  systems of equations by taking the sum of the squares of
polynomials.  

Following the  divide-and-conquer strategy, under  smoothness
and compactness assumptions, Safey El Din and Schost
\cite{din2017nearly} introduce a probabilistic roadmap algorithm,
with both output degree and running time  being polynomial in
$(nd)^{n\log(k)}$, where $k$ is the dimension of the considered
algebraic set. Note that the algorithm presented in
\cite{din2017nearly} exhibits a running time that is subquadratic in
the size of the output, and the associated complexity constants in the
exponent are explicitly provided. Recently, in
\cite{prebet2024computing}, Pr{\'e}bet, Safey El Din, and Schost
proved a new connectivity statement which generalizes the  one  in
\cite{din2017nearly}   to the unbounded case; a separate paper
to obtain an algorithm, with the same complexity as
\cite{din2017nearly}, for computing roadmaps is presented in \cite{prebet2024part2}.  

\paragraph*{Systems of symmetric polynomials.} The setup of this paper
will consider the connectivity question in the 
situation when the polynomials used to describe a semi-algebraic set
are symmetric of low degree (relative to the number of variables). It
has been observed by various works that the action of the symmetric
group can be particularly exploited in this situation. 

Timofte
\cite{timofte2003positivity} showed that for symmetric polynomials of
lower-degree, positivity can be inferred from positivity on
low-dimensional test sets, consisting of points with not more than
half-degree distinct coordinates. Building on this work,
\cite{riener2012degree,riener2016symmetric} and
\cite{riener2013exploiting} showed that this approach generally can be
used for algorithmic approaches to polynomial optimization problems,
yielding polynomial complexity. A similar approach was used by Basu
and Riener to compute the equivariant Betti numbers of symmetric
semi-algebraic sets \cite{basu2018equivariant},  the Euler Poincar\'e
characteristic \cite{basu2017efficient}, and the first $\ell$ Betti
numbers \cite{basu2022vandermonde} of such sets in a time which is
polynomial for a fixed degree. Recently, works by Faug{\`e}re, Labahn,
Safey El Din, Schost, and Vu \cite{faugere2020computing}, as well as
Labahn, Riener, Safey El Din, Schost, and Vu \cite{labahn2023faster}
have shown that also in the setup where the degree is apriori not
fixed, drastic reduction in complexity for certain algorithmic
questions in real algebraic geometry is possible.  
\subsection{Our results}
In this work, we demonstrate the feasibility of developing a polynomial time algorithm for determining connectivity within a symmetric semi-algebraic set, characterized by a constant number of symmetric polynomials with fixed degrees. Specifically, we present an algorithm with a computational complexity of $O(n^{d^2})$, designed to ascertain whether two points are equivariantly connected, presupposing that these points reside within the same fundamental domain of action. This assessment is equivalent to determining the connectivity of corresponding orbits in the orbit space. Our approach is informed by insights similar to those discussed in prior studies, particularly the feasibility of constructing an equivariant retraction from a fundamental domain in $S$ to its $d$-dimensional orbit boundary. For any two points $x$ and $y$, this retraction can be algebraically defined (as detailed in Theorem \ref{thm:Wel}), resulting in new points $x'$ and $y'$ situated within the $d$-dimensional boundary. Consequently, the initial question of connectivity is simplified to exploring the connectivity between $x'$ and $y'$ across the $d$-dimensional boundary of $S$ (as elucidated in Theorem \ref{thm:conection}). Given that this orbit boundary can be represented through the union of $d$-dimensional semi-algebraic sets, we devise a graph to encapsulate connectivity considerations (refer to Theorem \ref{thm:graph}). The primary algorithmic expenditure is attributed to the construction of this graph, offering a potentially more advantageous methodology compared to the more direct strategies employed in previous research \cite{basu2018equivariant,basu2017efficient,basu2022vandermonde}, which tend to increase degrees when considering such unions of $d$-dimensional sets.

\section{Preliminaries}\label{sec:prelim}
\subsection{Symmetric polynomials}
Throughout the article, we fix $n,d\in \N$, $d\leq n$ and denote by
$\R[\X] := \R[X_1,\dots,X_n]$ the real polynomial ring in $n$ variables. The \emph{basic closed semi-algebraic set} defined by polynomials $f_1, \dots, f_s$ in $\R[\X]$ is 
\[
\{x \in \R^n \, | \, f_i(x) \ge 0 {\rm \ for \ all} \ i = 1, \dots, s\}.
\] A semi-algebraic set is a set generated by a finite sequence of union, intersection and complement operations on basic semi-algebraic sets.

For $n \in \N_0$, we denote by $S_n$ the symmetric group of order $n!$. The symmetric group $S_n$ has a natural linear action on $\R^n$ by permuting the coordinates, which extends naturally to an action on $\R[\X]$. We say that a polynomial $f\in\R[\X]$ is symmetric if \[f(\X)=f(\sigma^{-1}(\X)) \text{ for all }\sigma\in\ S_n\] and we denote the $\R$-algebra of symmetric polynomials by
$\R[\X]^{S_n}$. For $1\leq i \leq n$ we let 
\[p_i:=X_1^i+\cdots+X_n^i\]
be the \emph{$i$-th power sum}. We further say that a semi-algebraic set $S$ is \emph{symmetric} if it is stable under the permutation action of $S_n$. Notice that clearly, a set described by symmetric polynomials is symmetric, but that also the converse is true, i.e., every symmetric set can be described by symmetric polynomials. Then the fundamental theorem of
symmetric polynomials states that every symmetric polynomial can be
uniquely written in terms of the first $n$ power sums, i.e., for every symmetric polynomial $f$ there is a unique polynomial $g\in\R[z_1,\ldots,z_n]$ such that
\[f=g(p_1,\ldots,p_n).\]Furthermore, due to the uniqueness of the representation it  follows directly that  for a symmetric polynomial of degree $d<n$ this  representation can not contain power sums of degree higher than $d$, i.e., we have: 

\begin{proposition}
Any symmetric polynomial $f\in \R[\X]^{S_n}$ of degree $d\le
n$, can be uniquely written as \[f=g(p_1,\dots,p_d),\]
where $g$ is a polynomial in $\R[Z_1,\dots,Z_d]$. \label{lm:deg_restrict}
\end{proposition} 
Note that there is no particular preference in using the power sum polynomials $p_i$. Indeed, any family of algebraic independent symmetric polynomials with degree sequence $\{1,\ldots, n\}$ yields the same results.

\subsection{Weyl chambers and the Vandermonde map}
We denote by $\W_c$ the cone defined by $X_1 \leq X_2 \leq \cdots \leq X_n$, the \emph{canonical Weyl chamber} of the action of $S_n$. This is a fundamental domain of this action. The walls of $\W_c$ contain points for which the above set of inequalities is not strict in some places, i.e., points $x=(x_1,\ldots,x_n)\in\W_c$ for which  $x_i=x_{i+1}$ for some $i\in\{1,\ldots,n\}$. More precisely:
\begin{definition}
\label{def:composition-order}
A sequence of positive integers $\lambda=(\lambda_1,\ldots,\lambda_\ell)$ with $|\lambda| := \sum_{i=1}^{\ell} \lambda_i = n$ is called a \emph{composition of $n$ into $\ell$ parts} and we call $\ell$ the \emph{length} of $\lambda$. 
Furthermore, we denote by $\Comp(n)$ the set of compositions of $n$ and by $\Comp(n,\ell)$ the set of compositions of $n$ of length $\ell$.
For $n \in \N$, and $\lambda = (\lambda_1,\ldots,\lambda_\ell) \in \Comp(n)$,
we denote by $\W_c^{\lambda}$ the subset of $\W_c$ defined by,
\begin{multline*}X_1 = \cdots = X_{\lambda_1} \leq X_{\lambda_1+1} = \cdots = X_{\lambda_1+\lambda_2} \\
\leq \cdots \leq X_{\lambda_1+\cdots+\lambda_{\ell-1}+1} = \cdots = X_n.
\end{multline*}
For every $\lambda\in\Com(n, d)$ the set $\W_c^\lambda$ defines a $d$-dimensional face of the cone $\W_c$, and every face is obtained in this way through a composition.
We will denote by $L_\lambda$ the linear span of $\W_c^\lambda$. Note that 
\[
\dim L_\lambda = \dim \W_c^\lambda = \length(\lambda).
\]
Furthermore, for every $d$ we define the set of \emph{alternate odd} compositions $\CompMax(n,d)$ as 
\[\left\{\lambda=(\lambda_1,\ldots,\lambda_d) \in \Comp(n) \; \mid \;
\lambda_{2i+1} =1, 0 \leq i < d/2 \right\}.\]
\end{definition}

Let $n\in \N$, and  $\lambda,\mu \in \Comp(n)$ then we denote $\lambda \prec \mu$, if $\W^{\lambda}_c \supset \W^{\mu}_c$. Equivalently,  $\lambda  \prec  \mu$ if $\mu$ can be obtained from $\lambda$ by replacing some of the commas in $\lambda$ by $+$ signs. One finds that $\prec$ is a partial order on $\Comp(n)$ turning $\Comp(n)$ into a poset.
If $\lambda, \mu \in \Comp(n)$, then the smallest  composition that is bigger than $\lambda$ and $\mu$ is called  the \emph{join} of $\lambda$ and $\mu$. The face corresponding to the join of $\lambda$ and $\mu$ is $\W_c^\lambda\cap\W_c^\mu$.

For $x \in \W_c$, we denote by $\comp(x)$ the largest
      composition $\lambda$ of $n$, such that $x$ can be written as 
    \[ x = (\underbrace{z_1, \dots, z_1}_{\lambda_1\text{-times}},
    \underbrace{z_2, \dots, z_2}_{\lambda_2\text{-times}}, \dots,
    \underbrace{z_\ell, \dots, z_\ell}_{\lambda_\ell\text{-times}}), \] 
    where $z = (z_1, \dots, z_\ell) \in \R^\ell$, and $\ell$ is the length of
    $\lambda$. For a composition $\lambda = (\lambda_1, \dots, \lambda_\ell)
      \in \Comp(n)$ and $f \in \R[\X]$, we define the polynomial
   \[ f^{[\lambda]} := f(\underbrace{X_1, \dots,
      X_1}_{\lambda_1\text{-times}}, \underbrace{X_2, \dots,
      X_2}_{\lambda_2\text{-times}}, \dots, \underbrace{X_\ell, \dots,
      X_\ell}_{\lambda_\ell\text{-times}}). \]

Notice that $f^{[\lambda]}$ is a polynomial in $\ell$ variables and its image is exactly the image of $f$ when restricted to $\W_c^\lambda$.    

\begin{example}
Consider $n = 4$.   There exist eight compositions of $4$ listed as
$(4), (3,1),$ $(1, 3), (2,2), (2,1, 1), (1,2,1), (1,1,2),
(1,1,1,1)$. 
Notably, for  $\mu$ in the set  $\{(3,1),
(1,3), (4)\}$, we have  $(1,2,1) \prec \mu$. Furthermore, the join of
$(2,1,1)$ and $(1,1,1,1)$ is $(2,1,1)$. 

Given $x = (-1, 5 , 5, -3)$, we determine  $\comp(x) =
(1,2,1)$. Lastly, considering the polynomial $f = x_3^3 + x_1x_2 -
x_4$ and the composition $\lambda = (1,2,1)$, it follows that
$f^{[\lambda]} = x_1x_2 + x_2^3 - x_3$.   
\end{example} 

Finally, for $d\in\{1,\ldots,n\}$ we denote by
\[\abb{\nu_{n,d}}{\R^n}{\R^n}{x}{(p_1(x),\dots,p_d(x))}\] 
the called  $d$-\emph{Vandermonde map}. In the case when $d=n$ we will just use the term Vandermone map.
For $a=(a_1,\dots,a_d)\in \R^d$, the fiber of the $d$-Vandermonde map, i.e., the set 
\[V(a):= \{ x\in \R^n ~|~ p_1(x)=a_1,\dots,p_d(x) = a_d \} \] 
is called a \emph{Vandermonde variety} with respect to $a$. Notice that here in the case $n=d$ every non-empty Vandermonde variety is precisely the orbit of one point $x\in\R^n$.

The importance of the Vandermonde map was realized by Arnold, Giventhal, and Kostov in their works on hyperbolic polynomials. We note the main properties needed for our work in the following.

\begin{theorem}[\cite{arnold1986hyperbolic, givental1987moments,
    kostov1989geometric}] 
For $a\in \R^d$ and a Weyl-chamber $\W$, the Vandermonde map, when
restricted to $\W$, is a homeomorphism onto its image. Moreover, the
Vandermonde   variety $V(a)\cap \W$ is either contractible or empty. 
\end{theorem}

\subsection{Encoding of real algebraic points}
Our algorithm manipulates points obtained through subroutines, with
the outputs having coordinates represented as real algebraic
numbers. We encode such a point using univariate polynomials and a
Thom encoding. 

Let $\mathcal{F} \subseteq \R[\X]$ and $x \in \R^n$. A mapping $\zeta:
\mathcal{F} \rightarrow \{-1, 0, 1\}$ is called a {\em sign condition}
on $\mathcal{F}$. The {\em sign condition realized by} $\mathcal{F}$
on $x$ is defined as $\sign(\mathcal{F}, x): \mathcal{F} \rightarrow
\{-1, 0, 1\}$, where $f \mapsto \sign(f(x))$. We say $\mathcal{F}$
{\em realizes} $\zeta \in \{-1, 0, 1\}^{\mathcal{F}}$ if and only if
$\sign(\mathcal{F}, x) = \zeta$.

A {\em real univariate representation} representing $x \in \R^n$
consists of: 
\begin{itemize}
\item  a zero-dimensional parametrization $$\scrQ = (q(T), q_0(T),
  q_1(T), \dots, q_n(T)),$$ where $q, q_0, 
  \dots, q_n$ lie in $\Q[T]$ with gcd($q, q_0$) = 1, and  
\item a Thom encoding $\zeta$ representing an element $\vartheta \in
  \R$ such that  
\[ 
q(\vartheta) = 0 \quad \text{and} \quad x =
\left(\frac{q_1(\vartheta)}{q_0(\vartheta)}, \dots,
  \frac{q_n(\vartheta)}{q_0(\vartheta)}\right) \in \R^n. 
\] 
\end{itemize}
Let $\der(q) := \{q, q^{(1)}, q^{(2)}, \dots, q^{(\deg(q))}\}$ denote
a list of polynomials, where $q^{(i)}$, for $i > 0$, is the formal
$i$-th order derivative of $q$. A mapping $\zeta$ is called a {\em
  Thom encoding} of $\vartheta$ if $\der(q)$ realizes $\zeta$ on $x$
and $\zeta(q) = 0$. Note that distinct roots of $q$ in $\R$ correspond
to distinct Thom encodings \cite[Proposition~2.28]{BPR06}. We refer
readers to \cite[Chapters 2 and 12]{BPR06} for details about Thom
encodings and univariate representations.

\begin{example}
A real univariate representation  representing $x = (1/2,
3/2-\sqrt{2}/4)$ is $(\scrQ, \zeta)$, where 
\[
\scrQ = (q, q_0, q_1, q_2) = (T^2-2, 2T, T, 3T-1)
\]
and $\zeta = (1, 1)$ $($equivalently, $\zeta = (+, +) = (q' > 0 \wedge
q'' > 0)$$)$. 

Indeed, let $\vartheta = \sqrt{2}$ be a root of $q$. Then $x =
\left(\frac{q_1({\vartheta })}{q_0({\vartheta})},
  \frac{q_2({\vartheta })}{q_0({\vartheta })} \right)$. Moreover,
$q'(T) = 2T$ and $q''(T) = 2$, implying that $q'(\vartheta) > 0$ and
$q''(\vartheta) > 0$. Thus, $\zeta = (+, +)$ serves as a Thom encoding 
of $\vartheta$.   
\end{example}

Given a non-zero polynomial $q \in \R[T]$, we consider the routine
${\sf ThomEncoding}(q)$, which returns the ordered list of Thom encodings
of the roots of $q$ in $\R$. This routine can be executed using
\cite[Algorithm 10.14]{BPR06} with a complexity of $O(\delta^4
\log(\delta))$, where $\delta = \deg(q)$, involving arithmetic
operations in $\Q$.  

Let $q \in \R[T]$ be a polynomial of degree $\delta_q$ and $p \in
\R[T]$ be a polynomial of degree $\delta_p$. Additionally, let
$\thom(q)$ be the list of Thom encodings of the set of roots of $q$ in
$\R$. We denote by ${\sf Sign\_ThomEncoding}(q, p, \thom(q))$ the
routine, which, for every $\zeta \in$ Thom$(q)$ specifying the root
$\vartheta$ of $q$, returns the sign of $p(\vartheta)$. This routine
can be implemented using \cite[Algorithm 10.15]{BPR06} with a
complexity of $O(\delta_q^2(\delta_q\log(\delta_q) + \delta_p))$
arithmetic operations in $\Q$.

\subsection{Roadmaps and connectivity}

In general, a roadmap for a semi-algebraic set is a curve that has a
non-empty and connected intersection with all of its connected
components. 

Formally, let $S \subset \R^n$ be a semi-algebraic set. We denote by
$\pi$ the projection onto the $X_1$-axis, and for $x \in \R$, we set
$S_x = \{y \in \R^{n-1} : (x,y) \in S\}$. A roadmap for $S$ is a
semi-algebraic set $\RM(S)$ of dimension at most one contained in $S$,
which satisfies the following conditions: 
\begin{itemize}
\item[{\sf RM}$_1$.] For every semi-algebraic connected component $C$
  of $S$, $C \cap \RM(S)$ is a semi-algebraic connected set. 
\item[{\sf RM}$_2$.] For every $x \in \R$ and for every semi-algebraic
connected component $C'$ of $S_x$, $C' \cap \RM(S) \ne \emptyset$.
 \end{itemize} 
Let $\mathcal{M} \subset \R^n$ be a finite set of points. A roadmap
for $(S, \mathcal{M})$ is a semi-algebraic set $\RM(S, \mathcal{M})$
such that $\RM(S,\mathcal{M})$ is a roadmap of $S$, and $\mathcal{M}
\subset \RM(S,\mathcal{M})$. Roadmaps can be used to decide the
connectivity of semi-algebraic sets. We summarize below the results we
use in this paper.

\begin{theorem}[{\cite[Theorem 3]{basu2000computing}} and {\cite[Theorem 16.14]{BPR06}}]
\label{them:normal}
  Let $F$ be a sequence of polynomial in $\Q[\X]$ with the algebraic
  set $Z\subset \R^n$ defined by $F$ of dimension $k$. Consider $G =
  (g_1, \dots, g_s) \subset \K[\X]$. Let $S$ be a semi-algebraic subset of $Z$
  defined by   $G$. Let $d$ be a bound for the degrees of $F$ and
  $G$. Consider a finite set of points $\mathcal{M}
  \subset Z$, with cardinality $\delta$, and  described
  by real univariate representations of degree at most
  $d^{O(n)}$.  Then the following holds. 
\begin{itemize} 
 \item  There is an algorithm whose output is exactly one 
  point in every semi-algebraically connected component of $S$. The
  complexity of the algorithm is $s^{k+1}d^{O(n^2)}$  arithmetic
  operations in $\Q$.
\item  There
  exists an algorithm to  compute a roadmap for $(Z,
  \mathcal{M})$ using $\delta^{O(1)} d^{O(n^{1.5})}$ arithmetic
  operations in $\Q$.
\end{itemize}
 As a consequence, there exists an algorithm for deciding whether two
  given points,  described by real univariate representations of
  degree at most   $d^{O(n)}$, belong to the same connected component
  of $Z$ by using $s^{k+1}d^{O(n^2)}$  arithmetic operations in $\Q$.   
\end{theorem}
Note also that there are some other prior results concerning the
problem of computing roadmaps and deciding the connectivity of two
given points (see Subsection~\ref{subsec:result}), with better
complexities. However, the result stated in Theorem~\ref{them:normal}
is sufficient for our main algorithm in Section~\ref{sec:algo}. In the following sections, we use the notation ${\sf Normal\_Connect}(S, x, y)$ to refer to the algorithm described in this theorem. Given a semi-algebraic set $S$ and two real points $x$ and $y$, the algorithm returns {\bf true} if $x$ and $y$ belong to the same connected component of S; otherwise, it returns {\bf false}. 

\subsection{Graphs and connectivity of unions}
Our algorithm will relate the computation of connected components in unions of sets.

\begin{definition}
Let $G=(V,E)$ be a graph. Given two vertices $u,v\in V$  we say that they are \emph{connected} if there exists a sequence of vertices \(v_1, v_2, \ldots, v_k\) where \(v_1 = u\), \(v_k = v\), and \((v_i, v_{i+1}) \in E\) for all \(1 \leq i < k\). Moreover, a \emph{connected component} of  \(G\) is a maximal set of vertices \(C \subseteq V(G)\) such that for every pair of vertices \(v, u\) in \(C\), there is a path in \(G\) between \(v\) and \(u\). 
\end{definition}
Our main need for this concept is related to the following. Let $S_1,\ldots,S_k\subset\R^n$ be a collection of semi-algebraic sets. We are interested in understanding the connectivity of the union $\bigcup_{i=1}^k S_i$ and will use the following construction of a bipartite graph.

\begin{algorithm}
\caption{Construct Bipartite Graph from Semi-Algebraic Sets}\label{alog:G}
\begin{itemize}
    \item[{\bf In:}] a collection of semi-algebraic sets $S_1, \ldots, S_k$
\end{itemize}
\begin{itemize}
    \item[{\bf Out:}] a bipartite graph $G(A \cup B, E)$ where $A$ and $B$ are vertex sets representing connected components of the sets, and $E$ is a set of edges indicating connectivity between these components
\end{itemize}
\vspace{-0.1cm}
 \noindent\rule{8.5cm}{0.5pt}
 
\begin{enumerate}
    \item initialize $A, B, E$ as empty sets
    \item {\bf for} {each pair of indices $i, j$ with $1 \leq i, j \leq k$ and $i \neq j$} {\bf do}
    \begin{enumerate}
        \item     
        compute one point per connected component of  $S_{(i,j)} = S_i \setminus S_j$ and $S_{(j,i)} = S_j \setminus S_i$. Add the points together with the ordered pairs ${(i,j)}$ and ${(j,i})$ into $A$
        \item compute one point per connected component of $S^{i,j} = S_i \cap S_j$ and add each point to $B$ together with the tuple $(i,j)$
    \end{enumerate}
    \item {\bf for} each $u \in A$ {\bf do}
    \begin{enumerate}
        \item {\bf for} each $w \in B$ {\bf do}
        \begin{enumerate}
            \item  determine the semi-algebraic set $S_i$ to which both $u$ and $w$ belong by inspecting the associated ordered pair {$(i,j)$} and the associated tuple $\{u,w\}$
            \item {\bf if} {$u,w$ are in the same set $S_i$} {\bf then}
            \begin{enumerate}
                \item {\bf if} {\sc connect}({$u, w, S_i$}) {\bf then} add edge $(u, w)$ to $E$
            \end{enumerate}
        \end{enumerate}
    \end{enumerate}
    \item \textbf{return} $G(A \cup B, E)$
\end{enumerate}
\label{alg:ConstructGraph}
\end{algorithm}

This information on the pairs of connected components and their intersection contains enough information to deduce the connectivity of the union. Indeed, this can be seen directly via the usual Mayer-Vietoris spectral sequence, whose $(0.0)$-term of the spectral sequence stabilizes at the second page. From this sequence it follows that the
$0$-cohomology of the union $\bigcup_i S_i$  is isomorphic to the kernel of
\[
\bigoplus_i \mathbb{H}^0(S_i) \rightarrow \bigoplus_{i,j} \mathbb{H}^0(S_{ij}),\]
where the map is the usual generalized restriction map.
This observation together with the notion of connectivity in graphs directly yields.
\begin{theorem}\label{thm:graph}
Let $S_1,\ldots, S_k\subset\R^n$ be a finite collection of semi-algebraic sets each of dimension $n$ and described by polynomials of degree at most $d$. The connected components of  $S:=\bigcup_{i=1}^k S_i$ are in $1:1$ correspondence with the connected components of the graph $G:=G(A \cup B, E)$ as constructed in Algorithm \ref{alg:ConstructGraph}.  
\end{theorem}
\begin{proof}
The vertex set $A$ represents points from the differences $S_i \setminus S_j$, $B$ represents points from the intersections $S_i \cap S_j$, and $E$ represents the connectivity between these points in $A$ and $B$. To show that we have an injective pairing:  assume two distinct connected components in $S$, say $C_1$ and $C_2$, map to the same connected component in $G$. By construction, vertices in $A$ and $B$ correspond to unique connected components of $S_i \setminus S_j$ and $S_i \cap S_j$, respectively. Since $C_1$ and $C_2$ are distinct in $S$, they must yield distinct sets of points in $A$ and $B$, contradicting the assumption that they map to the same component in $G$. Hence, the mapping is injective. To show the surjectivity of the pairing:  consider a connected component in $G$, represented by a subset of vertices in $A \cup B$ and edges in $E$. By construction, each vertex in $A$ and $B$ corresponds to a connected component in the differences $S_i \setminus S_j$ or intersections $S_i \cap S_j$ of the semi-algebraic sets. Since edges in $E$ represent connectivity between these components, the corresponding points in $S$ form a connected subset. This subset corresponds to a connected component in $S$, ensuring that every component in $G$ maps back to a component in $S$, establishing surjectivity.
\end{proof}

\begin{theorem}[{\cite[Chapter 8]{jungnickel2005graphs}}]
Let \(G = (V, E)\) be an undirected graph. There exists an algorithm that computes all connected components of \(G\) with complexity \(O(|V| + |E|)\). Moreover the  space complexity is \(O(|V| + |E|)\).
\end{theorem}
In combination with Theorem \ref{thm:graph} we thus have a possibility of computing the connected components of a union.


\section{Connectivity for symmetric semi-algebraic sets}
\subsection{Restriction to subspaces}
Our approach rests on the following characterization of Arnold.

\begin{theorem}[{\cite[Theorem 7]{arnold1986hyperbolic}}]  \label{thm:arnold} 
\label{itemlabel:thm:arnold:b}
Let $a\in\R^d$, and suppose that the $d$-Vandermonde variety $V(a)\subset \R^n$ is non-singular. Then a point $x \in V(a)\cap \W_c$ is a maximizer of $p_{d+1}$ if and only if $x \in \W^\lambda_c$ for some $\lambda \in \CompMax(n,d)$.
\end{theorem}
\begin{definition}\label{def:Skd}
Given $n,d \in \N$, 
we denote  
\begin{eqnarray*}
\W_{c}^d &=& \bigcup_{\lambda \in \Comp(k,d)} \W_c^\lambda.
\end{eqnarray*} 
Furthermore, for  a semi-algebraic subset $S \subset \R^n$, we denote

\begin{equation*}
\label{eqn:def:Skd}
S_{n,d} = S \cap \W_{c}^d,
\end{equation*}
which we call the \emph{$d$-dimensional orbit boundary of $S$}.
Notice that if $d \geq n$, then $S_{n,d} = S \cap \W_{c}$.
\end{definition}
The first assertion in the following is an immediate consequence of Theorem \ref{thm:arnold} also the second follows more or less directly (see also \cite[Proposition 7]{basu2018equivariant} for details).
\begin{theorem} \label{thm:con}
Let $S\subset\R^n$ be a semi-algebraic set defined by polynomials of degree $d$. Then every connected component of $S\cap \W_c$ intersects with $S_{n,d}$. In fact, $S_{n,d}$ is a retraction of $S\cap\W_c$.
\end{theorem}

\subsection{Algorithmic approach}
\begin{lemma} \label{lm:min_connect}
Let $S \subset \R^n$ be a semi-algebraic set defined by symmetric
polynomials of degree at most $d$ with $d \leq n$ and let $x\in S$. Furthermore, let $a=(p_1(x),\dots,p_d(x))$ and assume that $x' \in \R^n$ is a minimizer of $p_{d+1}$ in $V(a)$.  Then, $x$ and $x'$ are connected. 
\end{lemma}

\begin{proof}
First, we observe that for every $x\in S$ and every $a\in\R^d$, if
$x\in V(a)$, then $V(a)\subset S$. Indeed, let $z \in V(a)$. We will show that $z \in S$. By the definition of $V(a)$, all the Newton sums $p_1, \dots, p_d$ are constant on $V(a)$, that is, $p_i(z) = p_i(x)$ for all $i=1, \dots, d$. Moreover, since  $S$ is defined by
symmetric polynomials $F$ of degree at most $d$, by
Proposition~\ref{lm:deg_restrict}, each of these polynomials can be
expressed in terms of the first $d$ Newton sums  $p_1,\ldots,p_d$. That is 
$
F = G(p_1, \dots, p_d) 
$ for some  polynomials  $G$, which implies that $$F(x) = G(p_1(x), \dots, p_d(x)) = G(p_1(z), \dots, p_d(z)) = F(z).$$ Since $x\in S$, so is $z$.

Furthermore, since $V(a)$ is contractible and therefore connected, each $z \in V(a)$ belongs to the same connected
component as $x$. In particular, $x'$, which is in $V(a)$, and $x$ are in the same connected component.   
\end{proof}

\begin{theorem} \label{thm:Wel}
Let $V(a)$ be a non-empty Vandermonde variety, and let $\W_c$ be the
canonical Weyl-chamber. Then there is a unique minimizer of $p_{d+1}$
on $V(a)\cap \W_c\neq \emptyset$ with a multiplicity composition that
is bigger or equal to a composition $\lambda$ of length $d$ and 
$\lambda_{d}=\lambda_{d-2}=\lambda_{d-4}=\dots=1$.  

Furthermore, there exists an algorithm ${\sf MV}(a, \W_c)$ that takes
$a \in \R^d$ and the canonical Weyl chamber $\W_c$ as input and returns a
{real univariate representation} representing this  minimizer of 
$p_{d+1}$ on $V(a) \cap \W_c$ with a time complexity of \[
O\left( {n - \lceil d/2 \rceil
  -1 \choose \lfloor  {d/2} \rfloor  -1} \, d^{4d+1} \log(d)\right)  =
\softO\left({n^{d/2} d^d}\right) 
\] arithmetic operations in $\Q$. 
\end{theorem}

\begin{algorithm}
\caption{{\sf MV}$(a, \W_c)$}
\begin{itemize}
    \item[{\bf In:}] a point $a \in \R^d$ and the canonical Weyl-chamber $\W_c$ 
    \item[{\bf Out:}] a real univariate representation representing of the  minimizer of $p_{d+1}$ on $V(a) \cap \W_c$
\end{itemize}
\vspace{-0.1cm}
\noindent\rule{8.5cm}{0.5pt}
\begin{enumerate}
\item {\bf for} {$\lambda \in \Comp(n, d)$ with $\lambda_d =
  \lambda_{d-2} = \cdots  =1$} {\bf do}
\begin{enumerate}
\item compute polynomials $p^{[\lambda]} = (p_1^{[\lambda]}-a_1, \dots,
  p_d^{[\lambda]}-a_d)$
\item find a zero-dimensional parametrization $\scrQ$ of
  $V(p^{[\lambda]})$ 
\[\scrQ = (q(T), q_0(T),  q_1(T), \dots, q_d(T)) \in \Q[T]^{d+2}\]
\item define $\scrQ' \in \Q[T]^{n+2}$ with $$\scrQ' = \big(q(T), q_0(T), \underbrace{q_1(T), \dots,
      q_1(T)}_{\lambda_1\text{-times}}, \dots, \underbrace{q_d(T), \dots,
      q_d(T)}_{\lambda_d\text{-times}}\big) $$
\item compute $\thom(q) = {\sf ThomEncoding}(q)$ 
\item find  $\Sigma_0 = {\sf Sign\_ThomEncoding}(q, q_0, \thom(q))$ 
\item {\bf for} $i=1, \dots, d-1$ {\bf do} find $$\Sigma_i = {\sf Sign\_ThomEncoding}(q, q_i-q_{i+1},
  \thom(q))$$ 
\item  {\bf for} $i=1,\dots, d$   and $(\zeta_\vartheta,   \sigma_i) \in \Sigma_i$ {\bf do}
\begin{enumerate}
\item {\bf if} $\sigma_0 > 0 $  and $\sigma_i \le  0$ for all $i = 1, \dots,
  d-1$ \newline or  $\sigma_0 < 0$ and   $\sigma_i \ge  0$ for all $i
  = 1, \dots,  d-1$   
\begin{enumerate}
    \item {\bf return} $(\scrQ', \zeta_\vartheta)$
\end{enumerate} 
\end{enumerate}
\end{enumerate}
\end{enumerate}
\label{alg:MV}
\end{algorithm}

\begin{proof}
The first part of the theorem is obtained from the
  results in \cite{arnold1986hyperbolic,
    meguerditchian1992theorem}. For the second part, we consider the
  following polynomial optimization   $(\bf P)$:  $\min_{x \in
    V(a)}$, where $V(a)$ is the  Vandermonde variety   with respect
  to $a$. Since the Jacobian of the map $\big(p_1(\X)-a_1, \dots, 
p_d(\X)-a_d, p_{d+1}(\X)\big)$ is \[ 
J = \begin{pmatrix}
1 & 1 & \cdots & 1 \\
2X_1 & 2X_2 & \cdots & 2X_n \\
\vdots & & & \vdots \\
(d+1)X_1^{d} & (d+1)X_2^{d} & \cdots & (d+1)X_n^{d}
\end{pmatrix},
\]  any $(d+1)$-minor of $J$ has the form
\[
\alpha \cdot \prod_{\substack{i_1 \le i < j \le i_{d+1},\\ (i_1, \dots,
    i_{d+1}) \subset 
     \{1, \dots, n\} }}(X_{i} - X_{j}),
\] where $\alpha \in \R_{\ne 0}$.  Therefore, for any point $x\in
\R^n$ with more than $d$ distinct coordinates, there exists a
$(d+1)$-minor of $J$ such that  this minor does not vanish at
$x$. This implies that any optimizer has at most $d$ distinct
coordinates. 

Let $\lambda = (\lambda_1, \dots, \lambda_d)$  be
a composition of  $n$ of length $d$ with $\lambda_d =
\lambda_{d-2} = \cdots =1$. Then the optimizers of the problem
$(\bf P)$ of type $\gamma$, where $\gamma$ is a composition of $n$ with
$\lambda \preceq \gamma$, have the form    
\[
x = (\underbrace{x_1, \dots,
      x_1}_{\lambda_1\text{-times}}, \underbrace{x_2, \dots,
      x_2}_{\lambda_2\text{-times}}, \dots, \underbrace{x_d, \dots,
      x_d}_{\lambda_d\text{-times}}), 
\] where  $\bar x = (x_1, \dots, x_d)$ is a solution of   
\begin{equation} \label{eq:zero}
p_1^{[\lambda]} - a_1 = \cdots = p_{d}^{[\lambda]} - a_d = 0. 
\end{equation} 
Since the composition has length $d$, the system \eqref{eq:zero} is
zero-dimensional in $d$ variables, and its solution set is given by a
zero-dimensional parametrization. We then check for the existence of a
real point $x = (x_1, \dots, x_d)$ in this solution set such that
$x_{i} \le x_{i+1}$ for $i=1, \dots, d-1$. To do this, we follow the
following steps.

Let $\scrQ = (q(T), q_0(T), q_1(T), \dots, q_d(T))$ be a
zero-dimensional parametrization for the solution set of
\eqref{eq:zero}, and $\thom(q)$ be the ordered list of Thom encodings
of the roots of $q$ in $\R$. Suppose $\zeta_\vartheta \in \thom(q)$ is
a Thom encoding of a real root $\vartheta$ of $q$. Then 
\[\bar x = \left(\frac{q_1(\vartheta)}{q_0(\vartheta)}, \dots,
  \frac{q_d(\vartheta)}{q_0(\vartheta)}\right) \in \R^d.\]
Since $\gcd(q, q_0) = 1$, $q_0$ does not vanish at $\vartheta$. If
$q_0(\vartheta) > 0$, then $x_1 \le x_2 \le \dots \le x_d$ if and only
if $q_i(\vartheta) - q_{i+1}(\vartheta) \le 0$ for all $i=1, \dots,
d-1$. Otherwise, when $q_0(\vartheta) < 0$, $x_1 \le x_2 \le \dots \le
x_d$ if and only if $q_i(\vartheta) - q_{i+1}(\vartheta) \ge 0$ for
all $i=1, \dots, d-1$. Thus, it is sufficient to test the signs of
polynomials $q_0(T)$ and $q_i(T) - q_{i+1}(T)$ at
$\vartheta$. Computing $\thom(q)$ can be done using the ${\sf
  ThomEncoding}$ procedure, while checking the signs of $q_0$ and $q_i 
- q_{i+1}$ is obtained using the ${\sf Sign\_ThomEncoding}$
subroutine. In summary, the ${\sf MV}$ algorithm is given in
Algorithm~\ref{alg:MV}. 

\smallskip
We finish the proof by analyzing the complexity of the {\sf MV
  algorithm}. First since $\lambda_d = \lambda_{d-2} =
\cdots = 1$, it is sufficient to consider compositions of $n - \lceil
d/2 \rceil$ of length $\lfloor d/2 \rfloor$, which is equal to ${n -
  \lceil d/2 \rceil -1 \choose \lfloor d/2\rfloor-1}$.

For a fixed composition $\lambda$ of length $d$, finding a
zero-dimensional parametrization for the solution set of
\eqref{eq:zero} requires $O(D^3 + dD^2)$ arithmetic operations in
$\Q$, where $D$ is the number of solutions of
\eqref{eq:zero}, which is also equal to the degree of $q(T)$. This can
be achieved using algorithms such as those presented in 
\cite{rouillier1999solving}. Since
the system \eqref{eq:zero} is zero-dimensional with $d$ polynomials
and $d$ variables, B\'ezout's Theorem gives a bound for $D \le
d^d$. Therefore, the total complexity of finding a zero-dimensional
parametrization for the solution set of \eqref{eq:zero} is $O(d^{3d})$
arithmetic operations in $\Q$.

Finally, the subroutine ${\sf ThomEncoding}(q)$, which returns the
ordered list of Thom encodings of the roots of $q$ in $\R$, requires
$O(D^4 \log(D)) = O(d^{4d} \log(d^d)) = O(d^{4d+1} \log(d))$. Since
the degrees of $q_0$ and $q_i$ (for $i=1, \dots, d$) are at most that
of $q$, the procedure ${\sf Sign\_ThomEncoding}$ to determine the sign
of each of these polynomials at the real roots of $q$ is
$O(D^2(D\log(D) + D)) = O(d^{2d}(d^{d}\log(d^d) +d^d)) =
O(d^{3d+1}\log(d))$ operations in $\Q$. 

Thus, the total complexity of the $\sf MV$ algorithm is
 \[
O\left( {n - \lceil d/2 \rceil
  -1 \choose \lfloor  {d/2} \rfloor  -1} \, d^{4d+1} \log(d)\right)  =
\softO\left({n^{d/2} d^d}\right) 
\] arithmetic operations in $\Q$. This finishes our proof. 
\end{proof}

\section{The main algorithm}
\label{sec:algo}
We will now present the main result of the paper, but before doing so, we review and establish some notations. 
\begin{notation}
Let  $1\leq d \leq n$ be an integer, $z\in\W_c$ and set $a_1=p_1(z),\ldots, a_d=p_d(z)$. Then we will denote by $V_{d,z}$ the Vandermonde variety $V(a_1,\ldots,a_d)$. Furthermore, we denote by $z'$ the unique point in $V_{d,z}$ which exists according to Theorem \ref{thm:Wel}.
\end{notation}
\begin{theorem}\label{thm:conection}
Let $S \subset \R^n$ be a semi-algebraic set defined by symmetric
polynomials of degree at most $d$ with $d \leq n$ and take $x,y\in S\cap \W_c$. 
Then, $x$ and $y$ are connected within $S\cap \W_c$ if and only if the corresponding points $x'$ and $y'$ are connected within $S_{n,d}$.
\end{theorem}
\begin{proof}
By  Lemma \ref{lm:min_connect} $x$ and $y$ are both connected to $x'$ and $y'$ respectively.  Furthermore, since by Theorem \ref{thm:con} $S_{n,d}$ is a contraction of $S\cap W_{c}$ we have that $x'$ and $y'$ are connected within $S\cap W_{c}$ if and only if they are connected within $S_{n,d}.$
\end{proof}

We can now formulate an algorithm for checking equivariant connectivity of semi-algebraic sets defined by symmetric polynomials of degree at most $d$.

\begin{algorithm}
\caption{{\sf Connectivity\_Symmetric}($\f, (x, y)$)}\label{algorithm:special}
\begin{itemize}
    \item[{\bf In:}] {$\f = (f_1,\dots,f_s)\in \R[\X]^{S_n}$ of degrees at most $d$
  defining a semi-algebraic set $S$; two points  $x$  and $y$ in $\W_c$}
    \item[{\bf Out:}]  \textbf{true} if $x$ and $y$ are in the same connected
 of $S$;  otherwise return {\bf false}
\end{itemize}
\vspace{-0.1cm}
 \noindent\rule{8.5cm}{0.5pt}
\begin{enumerate}
    \item compute a list $L$ of alternate odd compositions of $n$ into $d$ parts.
    \item for $\lambda$ in $L$ compute $\f^{[\lambda]}:=(f_1^{[\lambda]},\dots,f_s^{[\lambda]})$.
    \item compute bipartite graph $G=(A\cup B, E)$ for $\bigcup_{\lambda \in L} S(\f^{[\lambda]})$.
    \item compute connected components of $G$.
    \item compute $a = (\nu_{n,d}(x))$ and $b = (\nu_{n,d}(y))$
    \item compute  $x' = {\sf MV}(a, \W_c)$ and $y' = {\sf MV}(b,
      \W_c)$ 
    \item find $\gamma = \comp(x')$ and $\eta = \comp(y')$
    \item find $v_x,v_y \in A$ with ${\sf Normal\_Connect}(\f^{[\gamma]}, (v_x,x'))={\bf true}$ and ${\sf Normal\_Connect}(\f^{[\eta]}, (v_y,y'))={\bf true}$.
    \item if $v_x$ and $v_y$ are in the same connect component of $G$ return {\bf true}, else return {\bf false}.
\end{enumerate}
\end{algorithm}

\begin{theorem}
Given symmetric polynomials $f_1,\dots, f_s \in \R[\X]^{S_n}$ defining a semi-algebraic set $S$ and two points $x,y \in S\cap \W_c$, the above Algorithm \ref{algorithm:special} correctly decides connectivity. Moreover, for fixed $d$ and $s$ its complexity can be estimated by $O(n^{d^2})$.
\end{theorem}
\begin{proof}
The correctness follows from Lemma \ref{lm:min_connect} in combination with Theorem \ref{thm:graph}. For the complexity of the algorithm we observe that the main cost is in building the graph $G$ with the help of algorithm \ref{alog:G}. Notice that this algorithm uses the algorithm of a road-map for the $d$-dimensional semi-algebraic sets defined by $\f^{[\lambda]}$. Therefore, following Theorem \ref{them:normal} each call to compute the points per connected component or to decide connectivity costs $s^{d+1}d^{O(d^2)}$ - which is constant by assumption. Moreover, the number of sets we consider is  \[\kappa(d)=O\left( {n - \lceil d/2 \rceil
  -1 \choose \lfloor  {d/2} \rfloor  -1} \, d^{4d+1} \log(d)\right)  =
\softO\left({n^{d/2} d^d}\right),\] which is the number of alternate odd compositions of $n$ with $d$ parts. Thus the first loop has length $\binom{\kappa(d)}{2}$. Every step of the loop a fixed complexity which is at most $3s^{d+1}d^{O(d^2)}$. For the second loop, note that we can bound the number connected components of each of the sets $S(\f^{[\lambda]})$ from above by $O((sd)^d)$. Therefore, the second loop has length at most $\kappa(d)^2 O((sd)^{2d})$ and each step has complexity at most $s^{d+1}d^{O(d^2)}$. Since $s$ and $d$ are fixed, we arrive at the announced complexity.  
\end{proof}


\section{Topics of further research}
In this paper, we have presented an efficient algorithm for determining the equivariant connectivity of two points within a semi-algebraic set $S$, focusing on symmetric polynomials of low degrees relative to the number of variables to obtain a polynomial complexity.  Looking ahead, we aim to broaden the scope of our research to encompass general connectivity issues, such as connecting points across different Weyl chambers. A promising strategy for achieving this involves leveraging the group action on connected components to enhance the underlying graph of our algorithm through this action.

Furthermore, we are interested in extending our methodology to address general symmetric semi-algebraic sets comprehensively. An intriguing avenue for future exploration is the application of our current ideas to enhance the computation of the first $\ell$ homology groups, potentially surpassing the findings in \cite{basu2022vandermonde}. A pivotal aspect of our approach has been the exploitation of the Vandermonde variety's contractibility and the insight that controlling the degree of symmetric polynomials directly influences the power-sum polynomials utilized in their description. This observation suggests the possibility of relaxing the degree requirement for symmetric polynomials by adopting alternative representational forms, such as those outlined in \cite{riener2024linear}, which explore the utility of elementary symmetric polynomials.

Moreover, the complexity of our algorithm is predominantly determined by the number of faces, $\W_c^\lambda$, that must be considered. Recent results by Lien and Schabert \cite{lienschabert} indicate potential methods for reducing the requisite number of faces.

Additionally, the recent works of Faugère et al. \cite{faugere2020computing} and Labahn et al. \cite{labahn2023faster} highlight the feasibility of employing symmetry at the computation level of critical points, a crucial element in the algorithms for generating road maps. This approach suggests that reducing complexity might also be achievable in scenarios where $d\ge n$, thereby broadening the applicability and efficiency of our proposed algorithm.

\balance

\bibliographystyle{plain} \bibliography{biblio}

\begin{thebibliography}{10}

\bibitem{arnold1986hyperbolic}
V.~I. Arnold.
\newblock Hyperbolic polynomials and {V}andermonde mappings.
\newblock {\em Funktsional'nyi Analiz i ego Prilozheniya}, 20(2):52--53, 1986.

\bibitem{basu2000computing}
S.~Basu, R.~Pollack, and M.-F. Roy.
\newblock Computing roadmaps of semi-algebraic sets on a variety.
\newblock {\em Journal of the American Mathematical Society}, 13(1):55--82, 2000.

\bibitem{BPR06}
S.~Basu, R.~Pollack, and M.-F. Roy.
\newblock {\em Algorithms in real algebraic geometry}.
\newblock Algorithms and computation in Mathematics. Springer-Verlag, second edition edition, 2006.
\newblock online version (2008).

\bibitem{basu2017efficient}
S.~Basu and C.~Riener.
\newblock Efficient algorithms for computing the euler-poincar{\'e} characteristic of symmetric semi-algebraic sets.
\newblock In {\em Ordered Algebraic Structures and Related Topics: International Conference on Ordered Algebraic Structures and Related Topics, October 12--16, 2015, Centre International de Rencontres Math{\'e}matiques (CIRM), Luminy, France}, volume 697, pages 53--81. American Mathematical Soc. Providence, Rhode Island, 2017.

\bibitem{basu2018equivariant}
S.~Basu and C.~Riener.
\newblock On the equivariant {B}etti numbers of symmetric definable sets: vanishing, bounds and algorithms.
\newblock {\em Selecta Mathematica}, 24:3241--3281, 2018.

\bibitem{basu2022vandermonde}
S.~Basu and C.~Riener.
\newblock Vandermonde varieties, mirrored spaces, and the cohomology of symmetric semi-algebraic sets.
\newblock {\em Foundations of Computational Mathematics}, 22(5):1395--1462, 2022.

\bibitem{basu2014divide}
S.~Basu and M.-F. Roy.
\newblock Divide and conquer roadmap for algebraic sets.
\newblock {\em Discrete \& Computational Geometry}, 52:278--343, 2014.

\bibitem{basu2014baby}
S.~Basu, M.-F. Roy, M.~Safey El~Din, and {\'E}.~Schost.
\newblock A baby step--giant step roadmap algorithm for general algebraic sets.
\newblock {\em Foundations of Computational Mathematics}, 14:1117--1172, 2014.

\bibitem{Canny}
J.~Canny.
\newblock {\em The complexity of robot motion planning}.
\newblock MIT Press, 1987.

\bibitem{canny1993computing}
J.~Canny.
\newblock Computing roadmaps of general semi-algebraic sets.
\newblock {\em The Computer Journal}, 36(5):504--514, 1993.

\bibitem{canny1992finding}
J.~Canny, D.~Yu Grigor'ev, and N.~N. Vorobjov.
\newblock Finding connected components of a semialgebraic set in subexponential time.
\newblock {\em Applicable Algebra in Engineering, Communication and Computing}, 2(4):217--238, 1992.

\bibitem{capco2023positive}
J.~Capco, M.~Safel El~Din, and J.~Schicho.
\newblock Positive dimensional parametric polynomial systems, connectivity queries and applications in robotics.
\newblock {\em Journal of Symbolic Computation}, 115:320--345, 2023.

\bibitem{capco2020robots}
J.~Capco, M.~Safey El~Din, and J.~Schicho.
\newblock Robots, computer algebra and eight connected components.
\newblock In {\em Proceedings of the 45th International Symposium on Symbolic and Algebraic Computation}, pages 62--69, 2020.

\bibitem{faugere2020computing}
J.-C. Faug{\`e}re, G.~Labahn, M.~{Safey El Din}, {\'E}.~Schost, and T.~X. Vu.
\newblock Computing critical points for invariant algebraic systems.
\newblock {\em Journal of Symbolic Computation}, 116:365--399, 2023.

\bibitem{givental1987moments}
A.~B. Givental.
\newblock Moments of random variables and the equivariant morse lemma.
\newblock {\em Russian Mathematical Surveys}, 42(2):275--276, 1987.

\bibitem{gournay1993construction}
L.~Gournay and J.-J. Risler.
\newblock Construction of roadmaps in semi-algebraic sets.
\newblock {\em Applicable Algebra in Engineering, Communication and Computing}, 4(4):239--252, 1993.

\bibitem{grigor1992counting}
D.~Yu Grigor'ev and N.~Vorobjov.
\newblock Counting connected components of a semialgebraic set in subexponential time.
\newblock {\em Computational Complexity}, 2:133--186, 1992.

\bibitem{heintz1994single}
J.~Heintz, M.-F. Roy, and P.~Solern{\'o}.
\newblock Single exponential path finding in semi-algebraic sets, part {II}: The general case.
\newblock In {\em Algebraic Geometry and its Applications: Collections of Papers from Shreeram S. Abhyankar’s 60th Birthday Conference}, pages 449--465. Springer, 1994.

\bibitem{iraji2014nuroa}
R.~Iraji and H.~Chitsaz.
\newblock Nuroa: A numerical roadmap algorithm.
\newblock In {\em 53rd IEEE Conference on Decision and Control}, pages 5359--5366. IEEE, 2014.

\bibitem{jungnickel2005graphs}
D.~Jungnickel and D.~Jungnickel.
\newblock {\em Graphs, networks and algorithms}, volume~3.
\newblock Springer, 2005.

\bibitem{kostov1989geometric}
V.~P. Kostov.
\newblock On the geometric properties of {V}andermonde's mapping and on the problem of moments.
\newblock {\em Proceedings of the Royal Society of Edinburgh Section A: Mathematics}, 112(3-4):203--211, 1989.

\bibitem{labahn2023faster}
G.~Labahn, C.~Riener, M.~Safey El~Din, {\'E}.~Schost, and T.~X. Vu.
\newblock Faster real root decision algorithm for symmetric polynomials.
\newblock In {\em Proceedings of the 2023 International Symposium on Symbolic and Algebraic Computation}, pages 452--460, 2023.

\bibitem{lienschabert}
A.~Lien and R.~Schabert.
\newblock Shellable slices of hyperbolic polynomials and the degree principle.
\newblock {\em arXiv}, (2402.05702), 2024.

\bibitem{meguerditchian1992theorem}
I.~Meguerditchian.
\newblock A theorem on the escape from the space of hyperbolic polynomials.
\newblock {\em Mathematische Zeitschrift}, 211:449--460, 1992.

\bibitem{prebet2024computing}
R.~Pr{\'e}bet, M.~Safey El~Din, and {\'E}.~Schost.
\newblock Computing roadmaps in unbounded smooth real algebraic sets {I}: connectivity results.
\newblock {\em Journal of Symbolic Computation}, 120:102234, 2024.

\bibitem{prebet2024part2}
R.~Pr{\'e}bet, M.~Safey El~Din, and {\'E}.~Schost.
\newblock Computing roadmaps in unbounded smooth real algebraic sets {II}: algorithm and complexity.
\newblock {\em arXiv preprint arXiv:2402.03111}, 2024.

\bibitem{riener2012degree}
C.~Riener.
\newblock On the degree and half-degree principle for symmetric polynomials.
\newblock {\em Journal of Pure and Applied Algebra}, 216(4):850--856, 2012.

\bibitem{riener2016symmetric}
C.~Riener.
\newblock Symmetric semi-algebraic sets and non-negativity of symmetric polynomials.
\newblock {\em Journal of Pure and Applied Algebra}, 220(8):2809--2815, 2016.

\bibitem{riener2024linear}
C.~Riener and R.~Schabert.
\newblock Linear slices of hyperbolic polynomials and positivity of symmetric polynomial functions.
\newblock {\em Journal of Pure and Applied Algebra}, 228(5):107552, 2024.

\bibitem{riener2013exploiting}
C.~Riener, T.~Theobald, L.~J. Andr{\'e}n, and J-B. Lasserre.
\newblock Exploiting symmetries in {SDP}-relaxations for polynomial optimization.
\newblock {\em Mathematics of Operations Research}, 38(1):122--141, 2013.

\bibitem{rouillier1999solving}
F.~Rouillier.
\newblock Solving zero-dimensional systems through the rational univariate representation.
\newblock {\em Applicable Algebra in Engineering, Communication and Computing}, 9(5):433--461, 1999.

\bibitem{safey2011baby}
M.~Safey El~Din and {\'E}.~Schost.
\newblock A baby steps/giant steps probabilistic algorithm for computing roadmaps in smooth bounded real hypersurface.
\newblock {\em Discrete \& Computational Geometry}, 45(1):181--220, 2011.

\bibitem{din2017nearly}
M.~Safey El~Din and {\'E}.~Schost.
\newblock A nearly optimal algorithm for deciding connectivity queries in smooth and bounded real algebraic sets.
\newblock {\em Journal of the ACM (JACM)}, 63(6):1--37, 2017.

\bibitem{schwartz1983piano}
J.~T. Schwartz and M.~Sharir.
\newblock On the “piano movers” problem. {II}. general techniques for computing topological properties of real algebraic manifolds.
\newblock {\em Advances in applied Mathematics}, 4(3):298--351, 1983.

\bibitem{timofte2003positivity}
V.~Timofte.
\newblock On the positivity of symmetric polynomial functions. part {I}: General results.
\newblock {\em Journal of Mathematical Analysis and Applications}, 284(1):174--190, 2003.

\end{thebibliography}

\end{document}